\documentclass{article}
\usepackage{tikz}
\usepackage{blkarray}
\usetikzlibrary{calc}
\usepackage{mathtools}
\usepackage{hyperref}
\usepackage{float}

\input{lafont-note.sty}

\title{Reversible $k$-valued logic circuits are finitely generated for
  odd $k$}

\author{\begin{tabular}{c}
    Peter Selinger
  \end{tabular}
}

\date{Dalhousie University}

\begin{document}
\maketitle

%----------------------------------------------------------------------
\section{Introduction}

Let $\Z_k=\s{0,\ldots,k-1}$, and let $\Ss[k]$ be the monoidal groupoid
whose objects are natural numbers $n$, and whose morphisms $f:k\to k$
are invertible functions $f:\Z_k^n\to \Z_k^n$.  In his 2003 paper
``Towards an algebraic theory of {Boolean} circuits''
{\cite[p.~298]{Lafont-2003}}, Lafont notes that, although $\Ss[k]$ is
not finitely generated when $k$ is even, it is finitely generated when
$k$ is odd. For the proof, he cites a private communication by me. The
purpose of this short note is to make the content of that
unpublished communication available. 

Since the proof referred to in Lafont's paper was never published,
others have proved the result independently. The earliest such
published proof that I am aware of is due to Boykett
{\cite{Boykett2015}}.

%----------------------------------------------------------------------
\section{Background: linear and affine invertible functions}

Let $\GL(\Z_k)$ and $\GA(\Z_k)$ be the monoidal subgroupoids of
$\Ss[k]$ consisting of invertible linear functions and invertible
affine functions, respectively. Here, we regard $\Z_k$ as a ring with
addition and multiplication. As usual, a function $f:\Z_k^n\to \Z_k^n$
is {\em linear} if $f(v+w) = f(v) + f(w)$ and $f(av)=af(v)$ for all
$v,w\in\Z_k^n$ and $a\in\Z_k$, and $f$ is {\em affine} if there exists
some $u\in\Z_k^n$ and a linear function $g$ such that $f(v) = u+g(v)$,
for all $v$. It is well-known from the theory of integer matrices that
$\GL(\Z_k)$ is finitely generated by the gates $D,U:2\to 2$, given by
$D(x,y)=(x,x+y)$ and $U(x,y)=(x+y,y)$, together with a gate
$H_a:1\to 1$ for each invertible element $a\in\Z_k$, given by
$H_a(x)=ax$.
If we moreover add a gate $v:1\to 1$ defined by $v(x)=x+1$, we
obtain a finite set of generators for $\GA(\Z_k)$. Lafont's notation
for these gates is shown in Figure~\ref{fig-gates}.
\begin{figure}[tp]
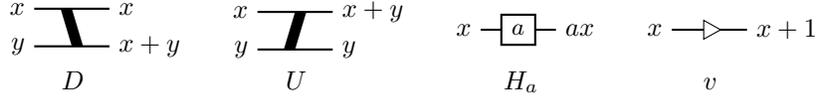

  \[
  \begin{array}{cccc}
    \m{\begin{qcircuit}[scale=0.5]
        \grid{2}{0,1};
        \leftlabel{$x$}{0,1};
        \leftlabel{$y$}{0,0};
        \lafontgate{1,1}{1,0};
        \rightlabel{$x$}{2,1};
        \rightlabel{$x+y$}{2,0};
      \end{qcircuit}}
    &\m{\begin{qcircuit}[scale=0.5]
        \grid{2}{0,1};
        \leftlabel{$x$}{0,1};
        \leftlabel{$y$}{0,0};
        \lafontgate{1,0}{1,1};
        \rightlabel{$x+y$}{2,1};
        \rightlabel{$y$}{2,0};
      \end{qcircuit}}
    &\m{\begin{qcircuit}[scale=0.5]
        \grid{2}{0};
        \leftlabel{$x$}{0,0};
        \rightlabel{$ax$}{2,0};
        \gate{$a$}{1,0};
      \end{qcircuit}}
    &\m{\begin{qcircuit}[scale=0.5]
        \grid{2}{0};
        \leftlabel{$x$}{0,0};
        \rightlabel{$x+1$}{2,0};
        \lighttriangle{1,0};
      \end{qcircuit}}
    \\
    D~~~~ & U~~~~ & H_a & v~~~~
  \end{array}
  \]
  \caption{Affine gates}\label{fig-gates}
\end{figure}

Note that the single transposition
$(0,\ldots,0,0)\leftrightarrow (0,\ldots,0,1)$ of $\Z_k^n$, together
with all invertible affine transformations, suffices to generate the
group of invertible functions on $\Z_k^n$. This is because by taking
affine conjugates, we can get all transpositions of the form
$(x_1,x_2,\ldots,x_{k-1},{x_k},x_{k+1},\ldots,x_n)\leftrightarrow
(x_1,x_2,\ldots,x_{k-1},{x_k}+1,x_{k+1},\ldots,x_n)$. These generate
all invertible functions because the permutation group $S(X)$ on any
set $X$ is generated by any set of transpositions $(ab)$ that form the
edges of a connected graph on $X$.

%----------------------------------------------------------------------
\section{The result}\label{sec-main}

\begin{proposition}
  If $k\geq 3$ is odd, then $\Ss[k]$ is finitely generated. In fact,
  it is generated by gates of arity $2$ and less.
\end{proposition}

\begin{proof}
  Let $S:1\to 1$ be the ``generalized negation'' gate defined by
  $S(0)=1$, $S(1)=0$, and $S(x)=x$ otherwise.  Consider the
  ``controlled negation'' gate, defined by
  \[
  S_{n+1} (x_1, \ldots, x_n, z) = (x_1, \ldots, x_n, z'),
  \]
  where $z'=S(z)$ if $x_1,\ldots,x_n=0$, and $z'=z$
  otherwise. Figure~\ref{fig-S} shows the notation we use for the
  controlled negation gate. The gate $S_{n}$ corresponds to a single
  transposition of $\Z_k^n$, exchanging $(0,\ldots,0,0)$ and
  $(0,\ldots,0,1)$. Therefore, the family of gates $S_{n}$ for
  $n\geq 1$, together with the affine transformations, generates
  $\Ss[k]$.
  \begin{figure}[tp]
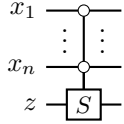

    \[
    \m{\begin{qcircuit}[scale=0.5]
        \grid{2}{0,1,2.5};
        \leftlabel{$x_1$}{0,2.5};
        \vdotslabel{0.5,1.25}
        \vdotslabel{1.5,1.25}
        \leftlabel{$x_n$}{0,1};
        \leftlabel{$z$}{0,0};
        \gencontrolled{\odot}{\gate{$S$}}{1,0}{1,2.5};
      \end{qcircuit}}
    \]
    \caption{The controlled negation gate $S_{n+1}$}\label{fig-S}
  \end{figure}

  Now consider the circuits $T$ and $U$ defined in Figure~\ref{fig-3}.
  \begin{figure}[tp]
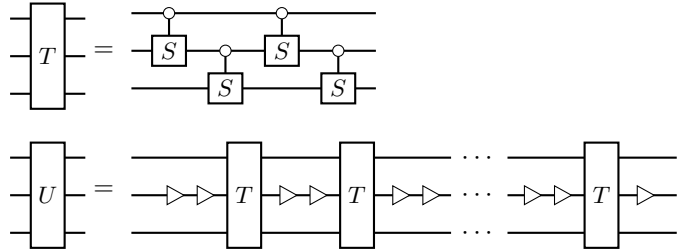

    \[
    \begin{array}{l}
      \m{\begin{qcircuit}[scale=0.5]
          \grid{2}{0,1,2};
          \biggate{$T$}{1,0}{1,2};
        \end{qcircuit}}
      =
      \m{\begin{qcircuit}[scale=0.5]
          \grid{6.5}{0,1,2};
          \gencontrolled{\odot}{\gate{$S$}}{1,1}{2};
          \gencontrolled{\odot}{\gate{$S$}}{2.5,0}{1};
          \gencontrolled{\odot}{\gate{$S$}}{4,1}{2};
          \gencontrolled{\odot}{\gate{$S$}}{5.5,0}{1};
        \end{qcircuit}}
      \\\\
    \m{\begin{qcircuit}[scale=0.5]
        \grid{2}{0,1,2};
        \biggate{$U$}{1,0}{1,2};
      \end{qcircuit}}
    =
    \m{\begin{qcircuit}[scale=0.5]
        \gridx{-1}{7.5}{0,1,2};
        \lighttriangle{0.1,1};
        \lighttriangle{0.9,1};
        \biggate{$T$}{2,0}{2,2};
        \lighttriangle{3.1,1};
        \lighttriangle{3.9,1};
        \biggate{$T$}{5,0}{5,2};
        \lighttriangle{6.1,1};
        \lighttriangle{6.9,1};
        \midlabel{$\cdots$}{8.25,0};
        \midlabel{$\cdots$}{8.25,1};
        \midlabel{$\cdots$}{8.25,2};
        \gridx{9}{13.5}{0,1,2};
        \lighttriangle{9.6,1};
        \lighttriangle{10.4,1};
        \biggate{$T$}{11.5,0}{11.5,2};
        \lighttriangle{12.6,1};
      \end{qcircuit}}
    \end{array}
    \]
    \caption{The circuits $T$ and $U$. Here, $U$ contains
      $\frac{k-1}{2}$ occurrences of $T$ and $k$ occurrences of the
      affine $v$-gate.}\label{fig-3}
  \end{figure}
  This circuit $T$ performs two transpositions
  $(0,0,0) \leftrightarrow (0,0,1)$ and
  $(0,1,0) \leftrightarrow (0,1,1)$ and is the identity elsewhere. The
  circuit $U$ uses $(k-1)/2$ affine conjugates of $T$ to perform the
  $k-1$ transpositions $(0,x,0) \leftrightarrow (0,x,1)$, for
  $x=1,\ldots,k-1$, and is the identity elsewhere.

  Finally, if we compose $U$ with $S_2$ as shown in
  Figure~\ref{fig-comp}, we obtain a single transposition
  $(0,0,0) \leftrightarrow (0,0,1)$, or in other words, $S_3$.
  \begin{figure}[tp]
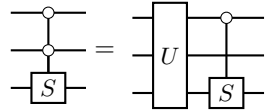

    \[
    \m{\begin{qcircuit}[scale=0.5]
        \grid{2}{0,1,2};
        \gencontrolled{\odot}{\gate{$S$}}{1,0}{1,2};
      \end{qcircuit}}
    =
    \m{\begin{qcircuit}[scale=0.5]
        \grid{3.5}{0,1,2};
        \biggate{$U$}{1,0}{1,2};
        \gencontrolled{\odot}{\gate{$S$}}{2.5,0}{2};
      \end{qcircuit}}
    \]
    \caption{Expressing $S_3$ in terms of $S_2$ and affine gates}\label{fig-comp}
  \end{figure}
  Therefore, $S_3$ is definable from $S_2$ in the presence of affine
  gates. By adding $n$ additional controls to every $S$-gate, it
  follows that $S_{n+3}$ is definable from $S_{n+2}$, for all
  $n\geq 0$. By induction, all $S_n$ for $n\geq 3$ are definable from
  $S_2$, and therefore $\Ss[k]$ is finitely generated when $k$ is odd
  and $k\geq 3$.
\end{proof}

%----------------------------------------------------------------------

\newpage
\bibliographystyle{unsrt}
\bibliography{lafont-note}

\begin{thebibliography}{1}

\bibitem{Lafont-2003}
Yves Lafont.
\newblock Towards an algebraic theory of {Boolean} circuits.
\newblock {\em Journal of Pure and Applied Algebra}, 184:257--310, 2003.

\bibitem{Boykett2015}
Tim Boykett.
\newblock Closed systems of invertible maps.
\newblock Available from \arxiv{1512.06813}, 2015.

\end{thebibliography}

%----------------------------------------------------------------------
\end{document}